\documentclass[aps,pre,groupedaddress,amssymb,amsmath,pra,floatfix,showpacs]{revtex4}
\usepackage{amssymb}
\usepackage{latexsym}
\usepackage{amsmath}
\usepackage{graphicx}

      \def\dC{{\mathbb C}}

      \def\dR{{\mathbb R}}

\def\bm\chi{\mbox{\boldmath$\chi$}}

\def\diag{{\rm diag\,}}

\let\xker=\ker \def\ker{{\xker\,}}

\def\supp{{\rm supp\,}}

\newtheorem{theorem}{Theorem}[section]
\newtheorem{proposition}[theorem]{Proposition}
\newtheorem{corollary}[theorem]{Corollary}

\newtheorem{example}[theorem]{Example}
\newtheorem{remark}[theorem]{Remark}

\newenvironment{proof}%
{\rm \trivlist \item[\hskip \labelsep{\bf Proof. }]}%
{\hspace*{\fill}$\Box$\endtrivlist}
{\rm \trivlist \item[\hskip \labelsep{\bf Proof}]}%
{\hspace*{\fill}$\Box$\endtrivlist}

\usepackage{color}

 \providecommand{\ii}{\mathrm i}

 \DeclareMathOperator{\const}{const}

\providecommand{\Real}{\mathbb{R}} \providecommand{\Comp}{\mathbb{C}}
\providecommand{\Nat}{\mathbb{N}} 
 \providecommand{\Zp}{\mathbb{Z}_+}

\providecommand{\eps}{\varepsilon} \providecommand{\dts}{,\dots,}
 \providecommand{\sbs}{\subseteq}
\providecommand{\To}{\Rightarrow}

\providecommand{\set}[1]{\left\{#1\right\}} \providecommand{\seq}[1]{\left<#1\right>}
\providecommand{\norm}[1]{\left\Vert#1\right\Vert} \numberwithin{equation}{section}
\providecommand{\Real}{\mathbb{R}} \providecommand{\Comp}{\mathbb{C}}
\providecommand{\Nat}{\mathbb{N}}

\begin{document}
\title{Truncations of a class of pseudo-Hermitian operators}
\author{Maxim Derevyagin${}^{*}$, Luca Perotti${}^{\dag}$, Micha\l{} Wojtylak${}^{\$}$}
\affiliation{
${}^{*}$ University of Mississippi, 
Department of Mathematics, 
Hume Hall 305, 
P. O. Box 1848, 
University, MS 38677-1848, USA. derevyagin.m@gmail.com}
\affiliation{${}^{\dag}$ Department of Physics, Texas Southern University, Houston, Texas 77004 USA. perottil@tsu.edu}
\affiliation{${}^{\$}$ Jagiellonian University, Faculty of Mathematics and Computer Science, \L ojasiewicza 6, 30-348 Krak\'ow. michal.wojtylak@gmail.com}
\date{\today}

\begin{abstract}
We consider the class of non-Hermitian operators represented by infinite tridiagonal matrices, selfadjoint in an indefinite inner product space with one negative square. We approximate them with their finite truncations. Both infinite and truncated matrices have  eigenvalues of nonpositive type: either a single one on the real axis or a couple of complex conjugate ones. As a tool to evaluate the reliability of the use of truncations in numerical simulations, we give bounds for the rate of convergence of their eigenvalues of nonpositive type. Numerical examples illustrate our results.


\end{abstract}


\maketitle

\section{Introduction}

The Hamiltonian $H$ of a physical system represents its energy, which is a real observable. It is therefore required that the expectation values of the quantum operator $H$ 
be real \cite{exep}. This can be guaranteed by imposing that $H$ 
be Hermitian, $H=H^\dagger$, as it is known that the spectrum of a Hermitian operator is real and its eigenvectors form a complete orthogonal set  \cite{dir}.

It is on the other hand known that Hermiticity is not a necessary condition for a real spectrum \cite{bend}: a large number of one-dimensional 
non-Hermitian potentials, both real and complex,
invariant under the simultaneous actions of the parity $P$ (space 
reflection) and time reflection $T$ operators \cite{note} have been found to admit energies that are real and discrete.

The matter is not a idle one, as non-Hermitian PT-invariant operators find applications in many areas of theoretical 
physics: ``optical" or ``average" potentials in nuclear physics \cite{boh}, quantum field theories \cite{Itz}, 
scattering problems \cite{Fesh}, localization-delocalization transitions in superconductors \cite{Fein}, defraction of 
atoms by standing light waves \cite{Berry}, as well as the study of solitons on a complex Toda lattice \cite{toda}.

Unfortunately, PT-invariance is neither necessary nor sufficient to ensure the reality of the spectrum; however, it has been conjectured \cite{bend} 
that PT invariant Hamiltonians possess real discrete eigenvalues if the PT 
symmetry is unbroken i.e. if the energy eigenstates are also eigenstates of the operator PT. When the PT-symmetry is 
broken and the Hamiltonian is real there instead are energy eigenvalues that are complex conjugate pairs. However, no 
general condition has been found for the breakdown of the PT-symmetry.

In this contest, it has been pointed out that a necessary, but not sufficient, condition for the spectrum to be real and discrete is the $\eta$-pseudo-Hermiticity, 
  $\eta H \eta^{-1} = H^\dagger$, of the Hamiltonian, where  $\eta$ is a Hermitian linear automorphism \cite{Nevanlinna}.  The property is also known as  selfadjointness in an indefinite inner product space, see \cite{bognar,langerio,GLR}. The eigenvectors of $H$ are in this case $\eta$-orthogonal, i.e. they are orthogonal 
according to the $\eta$-distorted inner-product $<\psi|\eta \psi>$.    

Several PT-symmetric potentials have been found to be P-pseudo-Hermitian \cite{Most} and classes of non-Hermitian Hamiltonians -both PT-symmetric and non-PT-symmetric- appear to be pseudo-Hermitian under $\eta = e^{-\theta p}$ where $\theta \in \Real$ and $p = −i \frac{d}{dx}$ , ($\hbar = 1$) is the momentum operator (the transformation generated by $\eta$ is an imaginary shift: $\eta x \eta^{-1} = x+i\theta$, $\eta p \eta^{-1} = p$) \cite{Ahmed}, or $\eta = e^{-\varphi(x)}$ where $\varphi(x)$ is a $C^1$ function of $x$ (the transformation is a complex gauge-like one) \cite{Ahmed2}.

It is thus still a case by case procedure to check whether the eigenvalues of an operator are all real. This does not usually cause big practical problems when dealing with a single operator. The situation changes when we have to consider classes of operators. Procedures have been developed for families of operators acting on spaces with finite bases, see e.g. Ref. \cite{Weig} whose author considers a one parameter family of PT-symmetric matrices $M(\varepsilon)$, with a perturbation parameter $\varepsilon \in \Real$ which destroys Hermiticity while it respects PT-invariance.

Here we consider the case when for numerical simulations an operator $H$ acting on a space with an infinite basis needs to be truncated and study the rate of convergence to their asymptotic value of those eigenvalues that for truncated matrices may happen to be non-real. The  operator $H=H_{[0,\infty)}$ is given by a non-symmetric Jacobi matrix 
\begin{equation}\label{H0infty}
H_{[0,\infty)}=
\begin{pmatrix}
  a_0 & -b_0 &  &  \\
  b_0 & a_1 & b_1 &  \\
      & b_1 & a_2 & b_2 \\
      &      &   b_2  & a_3 & \ddots\\
      &    &  & \ddots & \ddots \\
\end{pmatrix}, 
\end{equation}
with bounded, real sequences $(a_j)_{j=0}^\infty$, $(b_j)_{j=0}^\infty$, the sequence $(b_j)_{j=0}^\infty$ being additionally strictly positive. Its finite truncations are of the form 
\begin{equation}\label{H0n}
H_{[0,n]}=
\begin{pmatrix}
  a_0 & -b_0 &  &  \\
  b_0 & a_1 & b_1 &  \\
           &   b_1  & a_2 & \ddots\\
          &  & \ddots & \ddots & b_{n-1}\\
  & & & b_{n-1} & a_n \\
\end{pmatrix},
\end{equation}
 and
\[
\eta=\diag(-1,1,1,\dots).
\]
Due to the fundamental theorem of Pontryagin \cite{pontriagin} each the operators $H_{[0,n]}$ has, generically, either a unique single eigenvalue $\lambda_n$ on the real axis with the eigenvector $f_n$ satisfying $\seq{f_n|\eta f_n}\leq 0$ or a single couple of complex conjugate eigenvalues $\lambda_n\in\Comp^+$, $\bar\lambda_n\in\Comp^-$ 
(to avoid confusion with the conventions used in some of the papers we quote, we note that here and in the following $\seq{x|y}$ always denotes the usual inner product --either in $\Comp^n$ or in $\ell^2$-- {\it linear with respect to the second variable}). 
The remaining part of the spectrum of $H_{[0,n]}$ is real. The same is true for the spectrum of the infinite matrix $H_{[0,\infty]}$ with the eigenvalue $\lambda_\infty$, see Section \ref{Prel} for details. The character of the convergence $\lambda_n\to\lambda_\infty$ is the main topic of our paper.

Our approach makes use of analytic representations of the function
\begin{equation}\label{mform}
m_{[0,\infty)}(z)=-\seq{e_0|(H_{[0,\infty)}-z)^{-1}e_0},
\end{equation}
which contains the full information about the spectrum of $H_{[0,\infty)}$ and of its $[n-1/n]$ Pad\'e approximants
\begin{equation}\label{mnform}
m_{[0,n]}(z)=-\seq{e_0|(H_{[0,n]}-z)^{-1}e_0}.
\end{equation}
In particular, $\lambda_n$ ($\lambda_\infty$) is a pole of $m_{[0,n]}(z)$ ($m_{[0,\infty]}$, respectively) and it can be characterized in analytic terms. Due to the locally uniform convergence of  $m_{[0,n]}$ to $m_{[0,\infty)}$  \cite{DD07}, the sequence $(\lambda_n)_{n=0}^\infty$ converges to  $\lambda_\infty$  (Corollary \ref{lconvergence}). Our main interest is the rate of this convergence. In particular we show its dependence of the placement of the eigenvalue $\lambda_\infty$ in $\Comp^+\cup\Real$. 

Our paper is organized as follows:  
\begin{itemize}
\item We give various analytic representations of the function \eqref{mform}, choosing in particular as our starting point
$$
\frac{-1}{m_{[0,\infty)}(z)}={a_0-z+b_0^2 \int_{t_3}^{t_4} \frac{d\mu(t)}{t-z}},
$$
where $\mu$ is some probability measure, see Theorem \ref{rep}.
\item In the case $\lambda_\infty\notin[t_3,t_4]$ we show that the convergence rate of $\lambda_n$ to $\lambda_\infty$ is exponential, with the base of the exponent increasing with the distance of $\lambda_\infty$ from $[t_3,t_4]$: see Theorem \ref{main-geom} below.
\item However, the $\lambda_n$'s tend to arrange themselves in branches spiraling into $\lambda_\infty$ and some of these branches can get trapped in the real axis for a number of iterations $n_0$ which can be relatively large when $\lambda_\infty$ is close to $[t_3,t_4]$. 
We show examples with different numbers of branches and compute an estimate for $n_0$ in Theorem \ref{N}.
\item In the case when $\lambda_\infty\in[t_3,t_4]$ we build an example to show that the convergence rate is in general worse than exponential.
\item In the concluding remarks we review the possible cases from the numerical point of view.
\end{itemize}

\section{Holomorphic representations of the $m$-function}\label{Prel}

We start with reviewing the spectral properties of the matrices $H_{[0,n]}$ and $H_{[0,\infty]}$. The matrix $H_{[0,n]}$ is selfadjoint in the indefinite inner-product space with the fundamental symmetry given by $\eta_n=[-1]\oplus I_n$ ($\eta_n$-pseudo-Hermitian). 
Consequently one of the following four possibilities applies:
\begin{itemize}
\item[(i)] $H_{[0,n]}$ is similar to a diagonal matrix with real entries, except two complex conjugate entries $\lambda_n\in\Comp^+$, $\bar\lambda_n\in\Comp^-$. The  eigenvectors $f_n, g_n$ corresponding to the eigenvalues $\lambda_n,\bar\lambda_n$ of $H_{[0,n]}$  satisfy 
 $\seq{f_n|\eta_n f_n}=\seq{g_n|\eta_n g_n}=0$, $\seq{f_n|\eta_n g_n}\neq0$.
 \item[(ii)] $H_{[0,n]}$ is similar to a diagonal matrix with real entries and there is precisely one eigenvalue $\lambda_n$
with the corresponding eigenvector $f_n$ satisfying $\seq{f_n|\eta_n f_n}<0$.
 \item[(iii)] $H_{[0,n]}$ is similar to a block-diagonal matrix with all the blocks real and one-dimensional,  except one block of the form 
 $$
 \begin{pmatrix}\lambda_n & 1\\ 0 &\lambda_n\end{pmatrix}\text{ with }\lambda_n\in\Real.
 $$
 The  eigenvector $f_n$ corresponding to the eigenvalue $\lambda_n$ of $H_{[0,n]}$  satisfies 
 $\seq{f_n|\eta_n f_n}=0$.
 \item[(iv)] $H_{[0,n]}$ is similar to a block-diagonal matrix with all the blocks real and one-dimensional,  except one block of the form 
 $$
 \begin{pmatrix}\lambda_n & 1 & 0\\ 0 &\lambda_n & 1\\ 0 & 0 & \lambda_n\end{pmatrix}\text{ with }\lambda_n\in\Real.
 $$
 The  eigenvector $f_n$ corresponding to the eigenvalue $\lambda_n$ of $H_{[0,n]}$  satisfies 
 $\seq{f_n|\eta_n f_n}=0$.
\end{itemize}
The cases (iii) and (iv) are non-generic, i.e. the set of all matrices $H_{[0,n]}$ for which one of them applies has measure zero. 
We refer the reader to \cite{GLR} for the full canonical form of matrices selfadjoint in indefinite inner-product spaces, which gives also a full description of the eigenvectors. We observe that the matrix $H_{[0,n]}$ may jump back and forth with $n$ among the four types above. 

The spectral properties of the infinite matrix $H_{[0,\infty)}$, understood as an operator on $\ell^2$, are more tricky: we refer the reader to \cite{JL85,JLT} for a full description and for canonical models. Here we note only that again there are essentially two possibilities: 
\begin{itemize}
\item[(i')] $H_{[0,\infty)}$ is similar to an orthogonal sum of a bounded selfadjoint operator in a Hilbert space and a diagonal matrix with  two complex conjugate entries $\lambda_n\in\Comp^+$, $\bar\lambda_n\in\Comp^-$. The  eigenvectors $f_n, g_n$ corresponding to the eigenvalues $\lambda_n,\bar\lambda_n$ of $H_{[0,\infty)}$  satisfy 
 $\seq{f_n|\eta_n f_n}=\seq{g_n|\eta_n g_n}=0$, $\seq{f_n|\eta_n g_n}\neq0$.
\item[(ii')] The spectrum of $H_{[0,\infty)}$ is real and $H_{[0,\infty)}$ has a (unique) real eigenvalue with the corresponding  eigenvector $f_\infty$ satisfying $\seq{f_\infty|\eta f_\infty}\leq0$. 
\end{itemize}
In the (ii') case the Jordan chain corresponding to $\lambda_\infty$ is again of length not greater than three. 

Now we specify the theory developed in  \cite{De03,DD04,DD07} to the case we are dealing with in the present work. Besides the matrices $H_{[0,\infty)}$ and $H_{[0,n]}$ defined in \eqref{H0infty} and \eqref{H0n}, we shall use the following truncations of the matrix $H_{[0,\infty)}$
\begin{equation}\label{H1n}
H_{[1,n]}=
\begin{pmatrix}
   a_1 & b_1 &  \\
        b_1   &   a_1  & \ddots \\
            & \ddots & \ddots & b_{n-1}\\
  & & b_{n-1} & a_n \\
\end{pmatrix},\qquad  n=1,2,\dots.
\end{equation}
Furthermore,  $H_{[1,\infty)}$ will stand for the infinite, symmetric Jacobi matrix with $(a_j)_{j=1}^\infty$ on the main and $(b_j)_{j=1}^\infty$ on the second diagonals. Similarly to \eqref{mform} and \eqref{mnform} we define the functions
\begin{equation}\label{m1}
m_{[1,n]}(z)=\seq{ e_1|(H_{[1,n]}-z)^{-1}e_1},\quad m_{[1,\infty)}(z)=\seq{e_1 |(H_{[1,\infty)}-z)^{-1}e_1}.
\end{equation}
Here $e_j$ stands for the $j$--th vector of the canonical basis of $\ell^2$. We call the functions appearing in \eqref{mform}, \eqref{mnform} and \eqref{m1} the \textit{$m$--functions} of the corresponding Jacobi matrix.  We refer the reader to \cite{GS} for a treatment of  $m$--functions of symmetric Jacobi matrices appearing in \eqref{m1}. The functions $m_{[1,n]}$ ($n\in\mathbb Z_+$) and $m_{[1,\infty)}$ are analytic in the open upper half-plane $\Comp^+$. The function  $m_{[0,\infty)}$ ($m_{[0,n]}$) is analytic in the upper half-plane, except $\lambda_\infty$ ($\lambda_n$, respectively).


Moreover, the Schur complement argument provides the  following crucial  relations \cite{DD04,GS}
\begin{equation}\label{m01}
m_{[0,n]}(z)=\frac1{z-a_0-b_0^2\ m_{[1,n]}(z)},\quad z\in\dC^+\setminus\{\lambda_n\},\ n\in\Zp,
\end{equation}
\begin{equation}\label{m01infty}
m_{[0,\infty)}(z)=\frac1{z-a_0-b_0^2\ m_{[1,\infty)}(z)},\quad z\in\dC^+\setminus\{\lambda_\infty\}.
\end{equation}
 Let us now recall the definition of  the class  $\mathcal{N}_1$. 
By $\mathcal{N}_1$ we define the set of generalized Nevanlinna functions with one negative square, that is the functions of one of the three forms
\begin{equation}\label{GenDes1}
\frac{(z-\alpha)(z-\overline{\alpha})}{(z-\beta)(z-\overline{\beta})}\varphi(z),
\end{equation}
\begin{equation}\label{GenDes2}
 \frac1{(z-\beta)(z-\overline{\beta})}\varphi(z),
\end{equation}
\begin{equation}\label{GenDes3}
 {(z-\alpha)(z-\overline{\alpha})}\varphi(z),
\end{equation}
where $\alpha$, $\beta$ are complex numbers and $\varphi$ is a Nevanllina function, i.e. $\varphi$ is holomorphic 
in $\dC_+$ and maps $\dC_+$ into $\dC^+\cup\Real$.  We refer the reader to   \cite{DHS1,DLLSh} for  equivalent definitions. 
Let us now formulate the theorem which fixes the subclass of $\mathcal{N}_1$ functions to be investigated in the present work:

\begin{theorem}\label{rep}
Let $m$ be a meromorphic function in the open upper half plane. The following conditions are equivalent.
\begin{itemize}
\item[(i)] There exist $\lambda_\infty\in\Comp^+\cup\Real$, $d\in\Real$ and a nontrivial Borel measure $\sigma$ having all moments  finite and supported on an interval $[t_1,t_2]$ such that
\begin{equation}\label{n1model1}
 m(z)= \frac{1}{(z-\lambda_\infty)(z-\overline{\lambda}_\infty)}\left(z+d+\int_{t_1}^{t_2}\frac{d\sigma(t)}{t-z}\right),
\end{equation}
\item[(ii)] There exist $a_0\in\Real$, $b_0>0$ and a nontrivial Borel probability measure $\mu$ having all moments finite and supported on an interval $[t_3,t_4]$ such that
\begin{equation}\label{n1model2}
\frac{-1}{m(z)}={a_0-z+b_0^2 \int_{t_3}^{t_4} \frac{d\mu(t)}{t-z}}
\end{equation}
\item[(iii)] There exist a matrix $H_{[0,\infty)}$ of the form \eqref{H0infty} 
with bounded entries $a_j\in\Real$, $b_j>0$, $j\in\mathbb{Z}_+$ such that
\begin{equation}\label{n1model4}
m(z)=m_{[0,\infty)}(z):=-\seq{e_0|(H_{[0,\infty)}-z)^{-1}e_0}.
\end{equation}
 \end{itemize}
Furthermore, the parameters $\lambda_\infty$ and $d$ and the measure $\sigma$ in (i), $a_0, b_0$ and $\mu$ in (ii),  and $a_j, b_j$, $j\in\mathbb{Z}_+$ in (iii) are uniquely determined;  the numbers $a_0$ and $b_0$ in statements (ii) and (iii) coincide and $\lambda_\infty$ from statement (i) is the (unique) eigenvalue of nonpositive type of the operator $H_{[0,\infty)}$ from statement (iii). 
 \end{theorem}

\begin{proof} The equivalence (ii)$\Leftrightarrow$(iii) is a consequence of equation \eqref{m01infty} and the classical theory which sets a correspondence between the functions $m_{[1,\infty)}$ and the Jacobi matrices $H_{[1,\infty)}$, see e.g. \cite{Ach61,GS}.

(iii)$\To$(i) Let $m=m_{[0,\infty)}$. 
From the construction in \cite{JL85} it follows that $m$  belongs to the class $\mathcal{N}_1$ and hence it has one of the forms \eqref{GenDes1}--\eqref{GenDes3}. 

Furthermore, expanding the resolvent  into a geometric series at infinity one sees that $m$ necessarily possesses  an asymptotic expansion at infinity
\begin{equation}\label{asymp}
m(z)=-\seq{e_0|(H_{[0,\infty)}-z)^{-1}e_0}=
\frac{1}{z}-\frac{s_{1}}{z^2}
-\dots-\frac{s_{2n}}{z^{2n+1}}-\cdots,
\end{equation}
with $s_j\in\Real$ ($j=1,2,\dots$)  (see \cite{DD04,DD07}). Comparing the forms \eqref{GenDes1}--\eqref{GenDes3} with  \eqref{asymp}, one gets by the  Hamburger--Nevanlinna theorem \cite{Ach61}  that the function $\varphi$ in \eqref{GenDes1}--\eqref{GenDes3}
can be represented in the form
\begin{equation}\label{phi}
\varphi(z)=z+d+\int_{t_1}^{t_2}\frac{d\sigma(t)}{t-z},
\end{equation}
where  $d\in\dR$, and $\sigma$ is a measure with all moments finite. 
Furthermore, comparing the expansions of \eqref{GenDes1}, \eqref{GenDes2} and \eqref{GenDes3} with \eqref{asymp} we can see that case \eqref{GenDes2} applies. In consequence,
\begin{equation}\label{n1model}
 m(z)= \frac{1}{(z-\lambda_\infty)(z-\overline{\lambda}_\infty)}\left(z+d+\int_{t_1}^{t_2}\frac{d\sigma(t)}{t-z}\right),
 \end{equation}
 Observe that the measure $\sigma$  cannot be a finitely supported (trivial) measure, since $b_j>0$ for all $j\in\Zp$ and in consequence neither $m_{[1,\infty)}$ nor $m_{[0,\infty)}$ are rational functions. The uniqueness of the parameters $\lambda_\infty$ and $d$  and of the measure $\sigma$ follows from the theory of $\mathcal{N}_1$ functions, see e.g. \cite{DHS1}. The fact that $\lambda_\infty$ is the unique eigenvalue of nonpositive type of $H_{[0,\infty)}$ follows e.g. from Ref. \cite{DD07} or \cite{JL85}.
  
 (i)$\To$(ii) Using the algorithm proposed in~\cite{De03} (see also~\cite{ADL07}), one can find that $m$ defined by~\eqref{n1model}
can be uniquely represented as
\begin{equation}
m(z)=\frac1{z-a_0-b_0^2\ m_1(z)},
\end{equation} 
where $
m_1(z)=\int_{t_3}^{t_4} \frac{d\mu(t)}{t-z},
$ is a Nevanlinna function with finite moments. 
\end{proof}
 
 \begin{remark}\label{gaps}
 
{\rm Already at this point we can say something about the influence of the the measure $\mu$ (spectrum of the matrix $H_{[1,\infty)}$) on the position of $\lambda_\infty$.
 
1) Conditions (i) and (ii) in Ref. \cite{JL85} tell us that $\lambda_\infty\in[t_3,t_4]$ if and only if
\begin{equation}\label{JL85}
\int_{t_3}^{t_4}|t-\lambda_\infty|^{-2} d\mu(t)\leq b_0^{-2},   \qquad  a_0-{\lambda_\infty} +b_0^2\int_{t_3}^{t_4} (t-{\lambda_\infty})^{-1}d\mu(t)=0.
\end{equation}
(In particular, since $b_0$ is strictly positive, for the first of these conditions to be true, $|t-\lambda_\infty|^{-2}$ needs to be a $\mu$-integrable function; the second condition is just the specialization of eq. \eqref{n1model2} to the case $z=\lambda_{\infty}$ and we introduce it here to fully characterize $\lambda_\infty$ itself). It follows that if the measure $\mu$ is sufficiently dense $\lambda_\infty$ cannot be on $[t_3,t_4]$, i.e. $\mu$ ``repels" the point $\lambda_\infty$. This is the case for Examples \ref{jumping} and \ref{jumping2} below.
 
2) If instead $\mu$ has gaps, these gaps tend to trap $\lambda_\infty$. To see this, let's assume that the spectrum of $H_{[1,\infty)}$ has a gap $(t_5,t_6)\subset [t_3,t_4]$ and that $a_0\in(t_5,t_6)$. From the definition of $H_{[0,\infty)}$, eq. \eqref{H0infty}, it's obvious that for $b_0=0$ the point $a_0$ is an eigenvalue of $H_{[0,\infty)}$ with the corresponding eigenvector $f$ satisfying $\seq{f|\eta f}\leq0$, i.e. $a_0=\lambda_\infty$. We now increase $b_0=0$; applying Rouch\'e's theorem to $-1/m_{[0,\infty)}$ and remembering that if $\lambda_\infty\notin\Real$ then $\bar\lambda_\infty$ is also an eigenvalue, we see that $\lambda_\infty$ moves  along the real axis until it  meets another part of the spectrum of $H_{[0,\infty)}$ (that is either an eigenvalue, or a part of the continuous spectrum, see Ref. \cite{SWW11,SWW14} for a detailed analysis of a similar problem). This means that if $\lambda_\infty \in (t_5,t_6)$, then for a small change of parameters $a_0,b_0$ the eigenvalue $\lambda_\infty$ stays in the gap. We shall see one such case in Example \ref{jumping3}.} 
\end{remark}
 
As already mentioned, $m_{[0,n]}$ is the $[n/n+1]$  Pad\'e approximant of $m_{[0,\infty)}$ and it is an $\mathcal{N}_1$ function for $n\geq 1$. Consequently it can be represented in one of the forms  \eqref{GenDes1}--\eqref{GenDes3}. As we have just done in Theorem \ref{rep} for $m_{[0,\infty)}$, 
one can specify this representation:
\begin{proposition}\label{rep2}
Each function $m_{[0,n]}$ $(n=1,2,\dots)$ admits a representation
\begin{equation}\label{n1modeln}
 m_{[0,n]}(z)= \frac{1}{(z-\lambda_n)(z-\overline{\lambda}_n)}\left(z+d_n+\int\frac{d\mu_n(t)}{t-z}\right),
 \end{equation}
with $\lambda_n\in\Comp^+$, $d_n\in\Real$ and $\mu_n$ a finitely supported measure. Furthermore, the parameters $\lambda_n$ and $d_n$ and the measure $\mu_n$ are  uniquely determined and $\lambda_n$ is the unique eigenvalue of nonpositive type of $H_{[0,n]}$.
\end{proposition}
The details of the proof of uniqueness can be found e.g. in \cite{DLLSh,L,DHS1}.   
The uniqueness in both Theorem \ref{rep} and Proposition \ref{rep2} guaranties that $\lambda_n$ ($n=1,2,\dots$) and $\lambda_\infty$ are properly defined. In the literature they are called the generalized poles of nonpositive type of the corresponding $\mathcal{N}_1$ function, see \cite{L}. Using the classical result saying that $m_{[1,n]}$ converges to $m_{[1,\infty]}$, see \cite{GS, Ach61,Simon98}, 
one can prove --via eq. \eqref{m01}-- the following convergence result, cf.  \cite{DD07}:

\begin{proposition}\label{Conv}
The functions $m_{[0,n]}$ $(n\in\Zp)$ converge  to $m_{[0,\infty)}$ as  $n\to\infty$, locally uniformly on $\dC_+\setminus([t_1,t_2]\cup\set{\lambda_\infty})$.
\end{proposition}

 
Further generalization to different types of $\eta$-selfadjoint Jacobi matrices can be found in \cite{DD07}. As a consequence we have the following corollary (cf. \cite{LaLuMa}):

\begin{corollary}\label{lconvergence} The pole $\lambda_n$ of $m_{[0,n]}$ converges  to the pole $\lambda_\infty$ of $m_{[0,\infty)}$ as $n\to\infty$.
\end{corollary}

\begin{proof}
If $\lambda\in\dC_+$ or is an isolated eigenvalue then this statement is a simple consequence of Proposition~\ref{Conv} and
the Rouch\'e theorem.

If instead $\lambda_\infty$ is real and is not an isolated eigenvalue then from Proposition~\ref{Conv} and
the Rouch\'e theorem we see that all the accumulation points of $\lambda_n$ lie in the spectrum of $H_{[0,\infty)}$, which is a compact set. 
Since both the functions
$m_{[0,n]}$ and $m_{[0,\infty)}$ belong to $\mathcal{N}_1$ and are of the type \eqref{GenDes2}, we have
\[
m_{[0,n]}(z)=\frac1{(z-\lambda_n)(z-\overline{\lambda}_n)}\varphi_n(z),\quad m_{[0,\infty)}(z)=\frac1{(z-\lambda_{\infty})(z-\overline{\lambda}_{\infty})}\varphi(z),
\]
where $\varphi_n$ and $\varphi$ are Nevanlinna functions and $r_n$ and $r$ are rational functions.
Suppose now that
there is a sub-sequence such that $\lambda_{n_k}\to\lambda_0\ne\lambda_\infty$. As a consequence, $\varphi_{n_k}$
should also converge to a Nevanlinna function $\varphi_0 \neq \varphi$ which contradicts the uniqueness of $\varphi$.
\end{proof}

\section{Convergence rates}

Now we are in a position to ask the principal question of this paper: 
\begin{itemize}
 \item[] \textit{What is the character of the convergence of $\lambda_n\to\lambda_\infty$? }
\end{itemize}
We mainly consider the situation when $\lambda_\infty$ is simple eigenvalue located outside the support of the measure 
$\mu$ in \eqref{n1model2}: we show a theoretical bound on the convergence rate and test it on examples.

\subsection{Theoretical results}
\hspace{2.6in}

In this section we consider the situation when $\lambda_\infty$ is a simple pole of $m$. Note that if $\lambda_\infty\in\Comp^+$, then it is necessarily a simple pole, due to Theorem \ref{rep} (i); moreover, there exists $n_0\in\Nat$ such that  $\lambda_n\in\Comp^+$ for $n>n_0$ (see Theorem \ref{N} below). 
If instead $\lambda_\infty$ is a simple real pole, we show that $\lambda_n$ is real for sufficiently large $n$.
We begin with a technical result, needed to prove our main theorem.

\begin{proposition}\label{3.2}
Let $m=m_{[0,\infty)}$ satisfy the (equivalent) conditions  (i), (ii), (iii) of Theorem \ref{rep}. If $\lambda_\infty\in\Comp\setminus[t_3,t_4]$ is a simple pole of $m$, then 
 $$
{\lambda_\infty-\lambda_n} =-\frac{b_0^2\left(m_{[1,\infty)}(\lambda_\infty)-m_{[1,n]}(\lambda_\infty)\right)}{1-b_0^2 m_{[1,\infty)}'(\lambda_\infty)}+\alpha_n,
$$
where  $\alpha_n$ is such that
$$
\frac{\alpha_n}{\sup_{x\in X} |m_{[1,\infty)}(x) - m_{[1,\infty)}(x)   |}\to 0, \qquad n\to \infty
$$
for any  disc  $X\sbs\Comp^+$ containing $\lambda_\infty$. If, additionally, $\lambda_\infty\in\Real$, then $\lambda_n\in\Real$ for sufficiently large $n$. 
\end{proposition}

\begin{proof} 
Let $X$ be an open disc, such that $\lambda_\infty\in X\sbs \Comp^+\setminus[t_3,t_4]$. Note that  for sufficiently large $n$ the functions 
\begin{equation}\label{mn}
m_n(z):=\frac1{m_{[0,n]}(z)}=z-a_0-b_0^2 m_{[1,n]}(z),
\end{equation}
as well as
\begin{equation}\label{minfty}
 m_\infty(z)=\frac1{m_{[0,\infty)}(z)}=z-a_0-b_0^2 m_{[1,\infty)}(z)
\end{equation}
belong to $\mathcal{C}(X)$, the complex Banach space of  continuous functions on $X$ with the supremum norm. Indeed, for sufficiently large $n$ the function $m_{[1,n]}(z)$ has no poles in $X\cap\Real$.
Also observe that $m_n$ converges to  $m_{\infty}$ in $\mathcal{C}(X)$, since $m_{[1,n]}(z)$ converges to $m_{[1,\infty)}(z)$ locally uniformly on $\Comp\setminus[t_3,t_4]$.  Consider the mapping
$$
F:\mathcal{C}(X)\times X \ni (m,x) \mapsto m(x) \in \Comp.
$$
As $\lambda_\infty\in\Comp^+$ is a simple pole of $m_{[0,\infty)}$, one has 
 $$
m_\infty'(\lambda_\infty)=(-1/m_{[0,\infty)})'(\lambda_\infty)\neq 0.
$$ 
Therefore, 
$$
\frac{\partial F}{\partial x} (m_{\infty},\lambda_\infty)=m'_{\infty}(\lambda)\neq 0,
$$
and we can apply the implicit function theorem in Banach spaces to the mapping $F$ (see e.g. \cite{krantz}). As a result we obtain in a neighborhood $U\times Y$ of $ (m_{\infty},\lambda_\infty)$ a  differentiable function $\xi:U\to Y$ such that
$$
\set{(m,x)\in   U\times Y: m(x)=0}= \set{(m,\xi(m)): m\in U}=0.
$$
We may take $Y$ so small that  $Y\sbs\Comp\setminus[t_3,t_4]$ and that 
$m_\infty$ has no other zeros in $Y$ except $\lambda_\infty$.
 Note that for sufficiently large $n$ one has $m_{n}\in U$. Hence, on one hand we have that for sufficiently large $n$
$$
 m_n(\xi(m_n))=F(m_n,\xi(m_n))=0,\quad m_n(x)\neq0,\ x\in U\setminus\set{\xi(m_n)}.
$$
On the other hand, $\lambda_n$ converges to $\lambda_\infty$ and $m_n(\lambda_n)=0$. Consequently, 
 $\lambda_n=\xi(m_n)$ for $n$ large enough. 
 
Now note that 
$$
\frac{ \partial x}{\partial m}(m_\infty)m= - \frac{\frac{\partial F}{\partial m}(m_\infty,x(m_\infty))m}{\frac{\partial F}{\partial x}(m_\infty,x(m_\infty)) }=-\frac{m(\lambda_\infty)}{m_\infty'(\lambda_\infty)}.
$$  
Furthermore,
\begin{eqnarray*}
{\lambda_\infty-\lambda_n}& =&\frac{ \partial x}{\partial m}(m_\infty){(m_\infty-m_n)}+\alpha(m_\infty-m_n)\\
&=&-\frac{m_\infty(\lambda_\infty)-m_n(\lambda_\infty)}{m_\infty'(\lambda_\infty)}+\alpha(m_\infty-m_n),
\end{eqnarray*}
where $\alpha(h_n)/\norm{ h_n}_{\mathcal{C}(X)}\to 0$ with $\norm{h_n}_{\mathcal{C}(X)}\to 0$. Set $h_n=m_n-m_\infty$ and
$$
\alpha_n=\alpha(m_\infty-m_n)=\lambda_\infty-\lambda_n +\frac{m_\infty(\lambda_\infty)-m_n(\lambda_\infty)}{m_\infty'(\lambda_\infty)}
$$ 
and note that the right-hand side of the above does not depend on the initial choice of the disc $X$. 
This finishes the proof of the first statement.

Now let $\lambda_\infty\in\Real\setminus[{t_3},{t_4}]$. By the locally uniform convergence of  $m_n$ to $m_\infty$ and by the Rouch\'e theorem there is a small disc  $Z$ with the center in $\lambda_\infty$, such that each function $m_n(z)$ has precisely one zero $z_n$ in $Z$. As $\lambda_n$ converges to $\lambda_\infty$ we must have $\lambda_n=z_n$ for large $n$. 
Therefore, $\lambda_n\in\Real$, otherwise $\bar\lambda_n\in Z$ is another zero of $m_n$ in $Z$, which is a contradiction.
\end{proof}

\begin{remark}\label{expl}
{\rm We are able now to prove the main result of our paper, Theorem \ref{main-geom}. First, though, we would like to stress that there are two equivalent ways of seeing it according to the objects we consider:

A first interpretation takes as its main object the tridiagonal matrix, presented here in a block form
$$
H_{[0,\infty)}=\left(
\begin{array}{c|ccc}
  a_0 & -b_0 & 0  & \cdots  \\ \hline
  b_0 &  &  &  \\
    0  &   & H_{[1,\infty)} &  \\
     \vdots &      &       & \\
      \end{array}\right);
$$
$\lambda_\infty$ is then the (unique) eigenvalue of nonpositive type of $H_{[0,\infty)}$, 
$\lambda_n$ is the unique eigenvalue of nonpositive type of the finite truncation $H_{[0,n]}$ of $H_{[0,\infty)}$, and the spectrum of $H_{[1,\infty)}$ is contained, by assumption, in $[t_3,t_4]$. 

A second interpretation considers instead a meromorphic function $m(z)$ having the representations \eqref{n1model1} and \eqref{n1model2} and its $[n-1/n]$ Pad\'e approximants $m_{[0,n]}$. The point $\lambda_\infty$ ($\lambda_n$) is then the unique pole of nonpositive type of $m(z)$ ($m_{[0,n]}$, respectively). 

In both settings Theorem \ref{main-geom} gives the convergence rate of $\lambda_n$ to $\lambda_\infty$, in terms of the ``distance" of $\lambda_\infty$ from the interval $[t_3,t_4]$:
the rate of convergence is at least exponential $\mathcal{O}(q^{-2n})$, where
the number $q$ is such that $\lambda_\infty$ lies on the ellipse with foci at $t_3,t_4$ and sum of its semi-axes equal to $({t_4}-{t_3})q/2$. Consequently, as confirmed by our numerical tests below, the convergence rate gets worse the larger is the eccentricity of said ellipse, i.e.: the convergence slows down when $\lambda_\infty$ is ``close" to the interval $[t_3,t_4]$.}
\end{remark}

\begin{theorem}\label{main-geom}
Let $\lambda_\infty$, $\lambda_n$, $t_3$, $t_4$ be as in Theorem \ref{rep} and Remark \ref{expl} above.  
If $\lambda_\infty\in\Comp\setminus[t_3,t_4]$ is a simple eigenvalue, then 
$$
\limsup_{n\to\infty}|{\lambda_\infty-\lambda_n} |^{1/n}\leq \frac{1}{q^2},
$$
 where $q=g+\sqrt{g^2-1}$, and
\begin{equation}\label{g}
 g=\frac{|\lambda_\infty-t_4|+|\lambda_\infty-t_3|}{t_4-t_3}>1 
 \end{equation}
is the reciprocal of the eccentricity of the ellipse through $\lambda_\infty$ with foci at $t_3,t_4$.
If, additionally, $\lambda_\infty\in\Real$, then $\lambda_n\in\Real$ for sufficiently large $n$. 
\end{theorem}

\begin{proof}
For $R>1$ let $L_R$ denote the closed set bounded by the ellipse with foci at ${t_3},{t_4}$ and the sum of its semi-axes equal to $({t_4}-{t_3})R/2$. 
Due to Theorem (2.6.2) in Ref. 
 \cite{niso}, one has
\begin{equation}\label{elipse}
 \limsup_{n\to\infty}\sup_{z\in \Comp \setminus  L_R} |m_{[1,n]}(z) - m_{[1,\infty)}(z) |^{1/n} \leq \frac 1{R^2}.
\end{equation}
Note that 
\begin{equation}\label{LR}
\lambda_\infty\notin L_R \iff  R< g +\sqrt{g^2-1}.
\end{equation}
Take any $R\in(1,g+\sqrt{g^2-1})$ and a small disc $X$, such that $\lambda_\infty\in X\sbs \Comp\setminus L_R$. From Proposition \ref{3.2} we obtain that
\begin{eqnarray*}
|{\lambda_\infty-\lambda_n} | &\leq & C_1  |m_{[1,n]}(\lambda_\infty) - m_{[1,\infty)}(\lambda_\infty) | +\alpha_n\\
&\leq& C_2 \sup_{z\in X}  |m_{[1,n]}(z) - m_{[1,\infty)}(z) |,
\end{eqnarray*}
where $C_1,C_2$ are constants, dependent on $H_{[0,\infty)}$ and $X$ only. As a consequence, 
 $$
\limsup_{n\to\infty}|{\lambda_\infty-\lambda_n} |^{1/n}\leq \frac{1}{R^2}.
$$
Letting $R\to g +\sqrt{g^2-1}$ finishes the proof.
\end{proof}

Note that Theorem \ref{main-geom} cannot be easily generalized to the case when $\lambda_\infty\in[t_3,t_4]\setminus\supp \mu$. For example, if the support of the measure $\mu$ consists of two disjoint intervals $[t_3,t_5]\cup[t_6,t_4]$ the estimate \eqref{elipse}, which was the key point in proving Theorem \ref{main-geom}, still holds only outside the ellipse with foci at $t_3,t_4$, the reason being that the union of the poles of the Pad\'e approximants of  $\int_{t_3}^{t_4} (t-z)^{-1}\mu(dt)$  may be dense in $[t_5,t_6]$,  see for instance Ref. \cite{Suetin02}. 

\subsection{Examples}
\hspace{2.6in}

In our examples we want to be able to choose the position of $\lambda_\infty$; it is therefore convenient to consider cases where it is possible to calculate it without resorting at first to truncated matrices. One way to is look for matrices $H_{[1,\infty)}$ such that 
the the corresponding functions $m_{[1,\infty)}$ have a closed, analytic form. This will allow us to use (numerical) root finding methods to calculate $\lambda_\infty$ solving the equation 
\begin{equation}\label{mathe}
 z-a_0-b_0^2\ m_{[1,\infty)}(z)=0.
\end{equation}
Remembering that
\begin{equation}\label{H-m}
m_{[1,\infty)}=\int_{t_3}^{t_4} \frac{d\mu(t)}{t-z} ,
\end{equation}
this reduces to finding a suitable measure $\mu(t)$ with finite support. 

The choice
\begin{equation}\label{H-m1}
d\mu=d\sigma_{\alpha,\beta}(t)=\chi_{[-1,1]}(t)\cdot \frac{(1-t)^\alpha(1+t)^\beta dt}{\int_{-1}^1  (1-s)^\alpha(1+s)^\beta ds  },
\end{equation}
where $\chi_{[-1,1]}(t)$ is the characteristic function of the interval $[-1,1]$, gives us the matrix $H_{[1,\infty)}$ corresponding to Jacobi polynomials with parameters $\alpha,\beta$. To construct $H_{[1,\infty)}$ we consider the Jacobi orthogonal polynomials, which form an orthogonal basis in $L^2(\sigma_{\alpha,\beta})$ and the multiplication operator $p\mapsto xp$ in $L^2(\sigma_{\alpha,\beta})$. The three term recurrence relation (4.5.1) and the normalization factors (4.3.3) in Ref. \cite{Szego} provide the tridiagonal representation  $H_{[1,\infty)}$ of the multiplication operator. The spectral theorem for selfadjoint operators guarantees that \eqref{H-m} with \eqref{H-m1} is satisfied.

The choice 
\begin{equation}\label{H-m2}
d\mu=\chi_{[-2,2]}(t)\cdot \frac{\sqrt{4-t^2}dt}{2\pi}
\end{equation}
instead gives the matrix $H_{[1,\infty)}$ with $a_j=0$, $b_j=1$, $j=1,2,\dots$ corresponding to the orthogonal polynomials associated with the Wigner semicircle measure.


To calculate $\lambda_n$ we use Matlab \cite{matlab}: first, we calculate all eigenvalues and eigenvectors of $H_{[0,n]}$; we then find $\lambda_n$ as the only eigenvalue of nonpositive type of $H_{[0,n]}$,
 i.e. the only eigenvalue which is either in the upper half-plane or is real with the corresponding eigenvector ${\bf x}$ satisfying 
$$
-|x_0|^2+\sum_{j=1}^n |x_j|^2\leq 0.
$$

A summary of the relevant results for all our examples can be found in Table \ref{Table}. Only graphs useful to our discussion are shown here; for the remaining cases quoted, pictures can be found in the supplementary files \cite{S}: the file names there refer to those given in Table \ref{Table}.

\begin{example}\label{jumping}
{\rm As our first example we take $\alpha=\beta=0$ in eq. \eqref{H-m1} so that
\begin{equation}\label{ana}
 m_{[1,\infty)}=\frac12\int_{-1}^1\frac{dt}{t-z}=\frac12(\log(1-z)-\log(-1-z)),
\end{equation}
  where the branch of the logarithm is chosen in such way that the above function is a Nevanlinna function. 
In this case $H_{[1,\infty)}$ corresponds, in the way described in the remarks above,  to the Legendre polynomials. 
We now vary the only remaining free parameters $a_0$ and  $b_0$ of $H_{[0,\infty)}$. Note in any case, due to Remark \ref{gaps}, we have $\lambda_\infty\notin[t_3,t_4]=[-1,1]$.
 
Let us start with $a_0=0.5$, $b_0=0.05$. It is immediately evident from  Figure \ref{Legendre_compl} 
that  the points $\lambda_n$ arrange themselves on three branches which spiral into $\lambda_\infty$. The point $\lambda_n$ jumps in a regular fashion from one branch to another: all $\lambda_n$'s with $n \mod 3=\const$ fall on the same branch. The two branches starting on the real axis leave it at $n=45$ and $n=190$ respectively, as can be seen from the plot of the imaginary part of $\lambda_n$ in Figure \ref{Legendre_ReIM}; the plots of the corresponding real parts have each a cusp at the same time, due to the inversion of the direction of motion of $\lambda_n$. The plot in Figure \ref{Legendre_conv} shows exponential convergence of $\lambda_n$ to $\lambda_\infty$ setting in soon after all three branches leave the real line.
The black reference line with slope $-2\log(q)$ represents the bound from Theorem \ref{main-geom}, the intercept is chosen so that the line is superimposed on the numerical data. 
In this example (as well as in subsequent examples with measures having no gaps) we see that  the estimate of the convergence rate in Theorem \ref{main-geom} is  sharp and consequently can not be improved in general. On the other hand,  the estimate is not always sharp, cf. Example \ref{jumping3}.

We have already mentioned that the points $\lambda_n$ arrange themselves regularly on branches. This behavior is common to most of our examples, except the cases when $\lambda_\infty$ is on the real axis, where because of Proposition \ref{3.2} there is  one branch only.    Since  the branches appear to approach $\lambda_\infty$ isotropically, as can be e.g. seen in the zoomed picture  of Figure \ref{Legendre_compl}, this suggests a convenient way of calculating the numerical value of $\lambda_\infty$ as the mean 
\begin{equation}\label{mean}
\lambda_\infty^{(N)}=k^{-1}\sum_{n=N-k+1}^N \lambda_n, 
\end{equation}
where $k$ denotes the number of branches and $N$ the last $n$ for which $\lambda_n$ is calculated.

The value of $\lambda_\infty$ obtained taking the average of the three last points (one for each branch) equals 
$0.4999 + 0.0039\ii$, which agrees well with the value $0.498631+0.00391397\ii$ obtained solving with Mathematica \cite{mathematica} eq. \eqref{mathe}.

We conclude noting that --here and in all our other examples-- the real axis forms a barrier for $\lambda_n$: it never crosses it and branches which touch it get stuck in it; this is clearly visible in the movie {\sf Legendre\_spirals\_a\_0.5.avi} that can be found among the supplementary files to this paper \cite{S}. This behavior is related to the symmetry of the spectrum with respect to the real axis. 
 
If we now keep $a_0=0.5$ constant and vary $b_0$ to assume the values $b_0=0.5, 0.1$, and $0.01$, in all cases the $\lambda_n$'s arrange themselves over three spiraling branches and the value $\lambda_\infty$ obtained from the average of the last three values of $\lambda_n$ agrees with the one calculated solving eq. \eqref{mathe} with Mathematica to the last digit shown in Table  \ref{Table}. In all cases convergence is exponential and the slope of logarithmic plots similar to that in Figure \ref{Legendre_conv} are in agreement with the value $-2\log(q)$ from Theorem \ref{main-geom} (see Table  \ref{Table}).  There are some case to case differences, but they do not affect the picture given above: for $b_0=0.5$ the three branches are not clearly visible, due to the very fast convergence of $\lambda_n$; for $b_0=0.1$ only one branch spends some time on the real axis, up to $n=46$; and for $b=0.01$ even after 1000 iterations one of the three branches has not yet left the real axis (for figures see the supplementary material \cite{S1}). 
 
%
%

Our last example with measure eq. \eqref{ana} is a case  when $\lambda_\infty\in\Real\setminus [t_3,t_4]$: if we take $a_0=1.001$ and $b_0=0.001$ we get $\lambda_\infty\simeq 1.001$ (Mathematica has problems solving eq. \eqref{mathe} in this case). The values of $\lambda_n$ are all real; we therefore have a single branch. Convergence is again exponential (for figures see the supplementary material \cite{S2}). It is instructive to compare this case to the case $a_0=0.5$,  $b=0.05$: the distance of $\lambda_\infty$ from the interval $[-1,1]$ is of the same order, but the convergence is much faster in the present case. This is due to the different eccentricities of the ellipses from Theorem  \ref{main-geom}: in the present case the eccentricity is smaller, and therefore $q$ is larger and in consequence the convergence rate is better.  

Finally, in Figure \ref{Ns} we summarize the convergence behavior when submatrix $H_{[1,\infty)}$ corresponds to measure eq. \eqref{ana}: we  vary the parameters $a_0\in(-1,1)$ and $b_0\in(0.05,0.5)$ and for each pair $(a_0,b_0)$ we compute the value $n_0$ where the last $\lambda_n$ branch leaves the real axis. Each pair $(a_0,b_0)$ determines a single point $\lambda_\infty$ in the upper half-plane, which we plot color coded according to $n_0$. The figure appears to have a fractal character for which we do not yet have an explanation but which we suspect to be related to the way the number of $\lambda_n$ branches varies with varying $a_0$. This can be seen in the movie {\sf Legendre\_spirals\_b\_0.01.avi} \cite{S} where we keep $b_0=0.01$ and vary $a_0$. We instead observe no change in the number of branches when varying $b_0$ at constant $a_0$ (see e.g. the movie {\sf Legendre\_spirals\_a\_0.5.avi}  \cite{S}).}
 
\end{example}

As we have already mentioned, branches are a common occurrence, not limited to the example just given. We give here a couple more examples.


\begin{example}\label{jumping2}
{\rm We now take $H_{[1,\infty)}$ corresponding to the orthogonal polynomials associated with the Wigner semicircle measure, i.e. the measure given by eq. \eqref{H-m2} and we choose $a_0=0.5$, $b_0=0.1$ as the remaining parameters for $H_{[0,\infty)}$. 
Note  due to Remark \ref{gaps} we again have $\lambda_\infty\notin[t_3,t_4]=[-2,2]$.

In Figure \ref{Wigner_compl} it is possible to see that the $\lambda_n$'s form twelve branches. Knowing this, we can use the recipe given above to calculate $\lambda_\infty$ as the mean eq. \eqref{mean} of the last twelve values of $\lambda_n$; the result agrees to the last digit shown in Table \ref{Table} with the value obtained solving eq. \eqref{mathe} with Mathematica.

Figure \ref{Wigner_conv} shows exponential convergence of $\lambda_n$ to $\lambda_\infty$ with the rate predicted by Theorem \ref{main-geom}; other plots concerning this example can be found as supplementary files \cite{S3}. 
 
The video {\sf Wigner\_spirals\_b\_0.01.avi} in \cite{S} shows the evolution of the $\lambda_n$ branches under the change of the parameter $a_0$. As in Example \ref{jumping}, it is evident that here too the number of branches changes with $a_0$. Looking at the first few frames of the movie it is also evident that there are branches that start off the real axis, hit it and --instead of continuing into $\Comp^-$-- get trapped in it moving horizontally for a number of $n$, and then leave it when the spiral reenters $\Comp^+$.}
\end{example}

\begin{example}\label{jumping3}
{\rm Here we present an example of a different nature. The matrix $H_{[1,\infty)}$ is constructed in such way, that its spectrum is a totally disconnected, Cantor-like set, see \cite{cantor} for details; the other parameters are $a_0=0.5$ and $b_0=0.1$. In this particular case we could count $27$  branches of $\lambda_n$; such a high number is hardly visible when plotting $\lambda_n$ in the complex plane but can be seen in the plot of $\log|\lambda_n-\lambda_\infty|$ in Figure \ref{Cantor_conv}: the $\lambda_n$'s form regular clusters of $27$ points.

What is particularly noteworthy is that --contrary to the previous examples-- the convergence rate is faster than the one predicted by Theorem \ref{main-geom}, as can be seen comparing the theoretical bound (black line) with the numerical points in Figure \ref{Cantor_conv}. This is probably due to the fact that the spectrum of  $H_{[1,\infty)}$ contains gaps.
Other pictures corresponding to this example (called {\sf Cantor\_a\_0.5\_b\_0.1\_xxx.eps}), as well as the video {\sf Cantor\_spirals\_b\_0.2.avi} showing the evolution of branches under the change of the parameter $a_0$, can be found in the supplementary files \cite{S}: again the number of branches changes with $a_0$; moreover intervals in $a_0$ where both $\lambda_\infty$ and all the $\lambda_n$ become real are evident and correspond to the gaps in $H_{[1,\infty)}$, see Remark \ref{gaps}.}
\end{example}

Although Theorem \ref{main-geom} proves the exponential rate of  convergence of $\lambda_n$ to $\lambda_\infty$, it does not say when this convergence starts manifesting itself: at least in theory we could have $|{\lambda_\infty-\lambda_n} |> q^{-2n}$ for some $n$. Our numerical tests indicate that the convergence is somehow ``better" than the above bound, as the asymptotic behavior is $K \cdot q^{-2n}$ with $K<1$. Even the curious initial behavior connected with the real axis we observed in several cases above, does not bring $|{\lambda_\infty-\lambda_n}|$ to exceed $q^{-2n}$. On the other hand it would be convenient to have an estimate of $n_0$ such that for $n>n_0$ the point $\lambda_n$ is sure to be outside the support $[t_3,t_4]$ of the measure $\mu$: if one or more of the $\lambda_n$'s branches is still on $[t_3,t_4]$, our justification for estimating $\lambda_\infty$ by the mean eq. \eqref{mean} is compromised. We give an upper bound for $n_0$ in Appendix \ref{app:uno}.


%
%


\subsection{$\lambda_\infty\in\Real$ is embedded in the spectrum of the representing measure}\label{realcase}
\hspace{2.6in}

We shall limit our investigation of the case when $\lambda_\infty\in[t_3,t_4]$ to an example where  $\lambda_\infty=t_3$, to show that convergence in this case is in general worse than exponential. 
To build our example, we start by recalling Remark \ref{gaps}.  The only example known to us of classical orthogonal polynomials satisfying condition \eqref{JL85} with $\lambda_\infty\in[t_3,t_4]$ are the Jacobi polynomials with parameters $\alpha\geq 2$ or $\beta\geq 2$ and with $\lambda_\infty=t_3$ or $\lambda_\infty=t_4$, respectively.

\begin{example}\label{realev}
{\rm We take $H_{[1,\infty)}$ such that
$$
m_{[1,\infty)}=\int _{-1}^1 \frac{\frac34(1+t)^2(1-t) dt}{t-z},
$$
i.e. we take $\alpha=2, \beta=1$ in eq. \eqref{H-m1}. We then choose $a_0=-5/3$, $b_0=\sqrt{2/3}$, so that $\lambda_\infty=-1$. 
This can be seen e.g. by  analyzing the matrix $H_{[0,\infty)}$ itself:
it results from \cite{JL85}  conditions (i) and (ii) that $-1$ is an algebraically simple eigenvalue of $H_{[0,\infty)}$ with the corresponding eigenvector $x$ satisfying
$
-|x_0|^2+\sum_{j=1}^\infty |x_j|^2=0,
$
i.e. $\lambda_\infty=-1$ is the unique eigenvalue of nonpositive type of $H_{[0,\infty)}$.

Furthermore, due to Theorem 2.2 ($p_{1s}$) of \cite{JL85},  $-1$ is a singular critical, algebraically simple eigenvalue (see \cite{JL85} for a classification of eigenvalues of nonpositive type).     
 
The log-log plot of $|\lambda_n-\lambda|$ Vs. $n$  in Figure \ref{Jacobicrit_conv} shows clearly that convergence in this case is only polynomial: $|\lambda_n-\lambda|\simeq n^{-2}$, as can be seen comparing the numerical data with the black reference line whose slope is $-2$. The plot of the real and imaginary part of $\lambda_n$ can be found in \cite{S} as the file {\sf Jacobi12crit\_ReIm.eps}.} 

\end{example}

\section{Concluding remarks}\label{Conclusions}

As already stated above, the main question we tried to answer here is: 
``Suppose we are only able to compute the spectrum of finite truncations of a pseudo-Hermitian tridiagonal matrix $H_{[0,\infty)}$; can we say something about
the spectrum of the full operator $H_{[0,\infty)}$? Most interestingly, can we predict if the spectrum of $H_{[0,\infty)}$ is real or not?"

Suppose we have performed $N$ iterations, increasing step by step the size of the truncated matrix $H_{[0,n]}$; one of the next four cases applies.

\begin{itemize}
\item $\lambda_n$ is real for all $n=1\dts N$, is either the minumum or the maximum of the spectrum of $H_{[0,n]}$ and is separated from the other eigenvalues. Then the limit point $\lambda_\infty$ is also a real single eigenvalue, separated from other eigenvalues of $H_{[0,\infty)}$. 
\item $\lambda_n$ is complex for all $n$ larger than a $n_0<N$. Then the limit point $\lambda_\infty$ is also a complex single eigenvalue. $\lambda_\infty$ itself can be evaluated by finding the number of $\lambda_n$ branches and then using eq. \eqref{mean}. 
\item $\lambda_n$  oscillates between the real line and the complex plane up to $n=N$, a common occurrence when $\lambda_\infty$ is very close to the support of the measure $\mu$.  
Then the situation is in principle unclear: the limit eigenvalue $\lambda_\infty$ might be a complex point, a real critical point, or a real point in a relatively small gap of the spectrum. Still, if the $\lambda_n$ branches can be found and not to many of them are still trapped on $\Real$ for $n\simeq N$, it is still possible to give a numerical evaluation $\lambda_\infty^{(N)}$ of $\lambda_\infty$ by a careful use of eq. \eqref{mean}. 
If the plot of $\log|\lambda_n-\lambda_\infty^{(N)}|$ is then approximately a straight line, this is another indication for $\lambda_\infty$ being a simple eigenvalue out of the support of $\mu$.
\item The sequence $\lambda_n$ converges to a point $\lambda_\infty\in\Real$, but the convergence is not exponential, which again can be seen by the study of the plot of  $\log|\lambda_n-\lambda_\infty|$. Then $\lambda_\infty$ is a critical point on the real line, embedded in the support of $\mu$. In view of the second equation of  \eqref{JL85},  this case seems to be non-generic, i.e. a small change of the entries of the matrix will lead to a different case. However,  we were able to clearly observe this case in a numerical simulation, which in our opinion is an argument for considering this possibility as well.
\end{itemize}



While, for sake of simplicity, we restricted ourselves to the case of matrices with a single eigenvalue of nonpositive type, we believe that the results derived above can be generalized to a wider  class of  operators, at least to those considered in Ref. \cite{DD07}.

In the context of random matrices \cite{Wojtylak12b} or Nevanlinna functions, our research can be viewed as concerning the problem of predicting whether or not $\lambda_\infty\in [t_3,t_4]$ by calculating a finite number of Pad\'e approximants.
A connection can also be found to Pad\'e approximation of the Z-transform, considered in \cite{BessisPerotti09}, where the real line is replaced by the unit circle.

\section*{Acknowledgments}

Micha\l{} Wojtylak gratefully acknowledges the financial assistance of the Alexander von Humboldt Foundation with a Research Grant for Experienced Scientists, carried out at TU Berlin and with a Return Home Scholarschip, carried out at Jagiellonian University, Krak\'ow.  
Maxim Derevyagin gratefully acknowledges the financial assistance of the European Research Council under the European Union Seventh Framework Programme (FP7/2007-2013)/ERC, grant agreement no. 259173.  Maxim Derevyagin and Micha\l{} Wojtylak are indebted to Professor Olga Holtz for her encouragement and  support.

\appendix
\section{Estimate of the maximum number of real $\lambda_n$, when $\lambda_\infty\in\Comp^+$}
\label{app:uno}

\begin{theorem}\label{N} Let $m=m_{[0,\infty)}$ be of the forms \eqref{n1model1} and \eqref{n1model2} with $\lambda_\infty\in\Comp^+$. 
Then for any $\eps\in(0,1)$ for
\begin{equation}\label{bigN}
n>N_{\eps}:=\frac{\log\left( \frac{8b_0^2g\max_{z\in g L_R} |m_{[0,\infty)}(z)|}{(g-1)^2(t_4-t_3)}\right)}{\log\left( 1+\eps\sqrt{1-g^{-2}} \right)}
\end{equation}
one has 
$$
\lambda_n\in\Comp\setminus L_{R(\eps)},
$$
where g is again given by eq. \eqref{g}, $R(\eps)=g+\eps \sqrt{g^2-1}$ and $L_{R(\eps)}$ again denotes the closed set bounded by the ellipse with foci at ${t_3},{t_4}$ and the sum of its semi-axes equal to $({t_4}-{t_3})R(\eps)/2$. In particular, if $n>N_\eps$ for some $\eps\in(0,1)$ then $\lambda_n\notin[t_3,t_4]$.
\end{theorem}

\begin{proof}
Consider the functions
$$
m_n(z):=\frac1{m_{[0,n]}(z)},\quad m_\infty(z)=\frac1{m_{[0,\infty)}(z)}.
$$
We now recall the estimate given in \cite{niso} as formula (6.10): for every $n\in\mathbb{Z}_+$ and for every $\delta\in (1,R)$ one has 
\begin{equation}\label{deltaR}
\max_{z\in\partial L_R} | m_{[1,\infty)}(z)-m_{[1,n]}(z)(z)|\leq \frac {8\delta}{(t_4-t_3)(\delta-1)^2}\left(\frac\delta R\right)^{2n}.
\end{equation}

Applying equation \eqref{deltaR} with $\delta=g$, $R=R(\eps)=g+\eps\sqrt{g^2-1}$ we get after elementary transformations of \eqref{bigN} that

\begin{eqnarray*}
 |m_n(z) - m_\infty(z)|& =&  b_0^2 |m_{[1,n]}(z) - m_{[1,\infty)}(z)|\\
 &\leq& \frac{8b_0^2g} {({t_4}-{t_3})(g-1)^2  } \left( \frac{g}{1+\eps\sqrt{g^2-1}}\right)^{2n}\\
 &\leq& \frac{1}{\max_{z\in\delta L_R} |m_{[0,\infty)}(z)|}
  = \min _{z\in\delta L_R} |m_{\infty}(z)|.
 \end{eqnarray*}
Hence, by the  the  Rouche theorem, $m_n$ and $m_\infty$ have the same number of zeros in $\bar\Comp\setminus L_R$.
However, $m_\infty$ has precisely two zeros in $\bar\Comp\setminus L_R$, namely $\lambda_\infty$ and $\bar\lambda_\infty$. As $\lambda_n$ is the only a zero of $m_n$ in the upper half-plane, we get $\lambda_n\in \Comp\setminus L_R$.

\end{proof}

If we now apply Theorem \ref{N} to Example \ref{jumping}, we get $N_{0.5}=52$ for $b_0=0.5$, $N_{0.5}=2155$ for $b_0=0.1$, 
$N_{0.5}=10917$ for $b_0=0.05$, and $N_{0.5}\simeq 4\cdot 10^5$ for $b_0=0.01$. Comparing with Example \ref{jumping}, where $n_0\simeq 1,46,190$, and $n_0 \gg 1000$ respectively, we see that the estimate $N_{0.5}$ is a far from tight upper bound whose ratio to the numeric value $n_0$ appears to grow with decreasing $b_0$.

\newpage
\begin{table}
\begin{tabular}{|l|c|c|c|c|c|c|c|c|c|}
\hline
  &$a_0$               & $b_0$& $\max n$   & $\lambda_\infty$      & $n_0$   & q         & $|\lambda_\infty-\lambda_n|$\\ \hline\hline
 Legendre & 0.5   & 0.5      & 50 & $0.4045 + 0.3064\ii$  &  1     & 1.3855 & $\simeq q^{-2n} $  \\ 
                & 0.5    & 0.1       & 200& $0.4946 + 0.0155\ii$ &   46   & 1.0180 & $\simeq3.5^{-1} q^{-2n}$    \\
                &0.5     & 0.05     &1000& $0.4986 + 0.0039\ii$ & 190    & 1.0045  & $\simeq4.6^{-1} q^{-2n}$     \\
                & 0.5    & 0.01     &1000& $0.4999 + 0.0002\ii$ & $>1000$ & 1.0002  &$ \simeq6.5^{-1} q^{-2n}$  \\
                & 1.001& 0.001   & 200 &1.0010                   &     ---           & 1.0456  &$\simeq12.5^{-1} q^{-2n}$  \\ \hline
 Wigner   & 0.5    &  0.1       &1000&$0.4975 + 0.0096\ii$&     175     &1.0050& $\simeq4^{-1}q^{-2n}$        \\ \hline
 Cantor & 0.5     & 0.1        &  400&  $0.5074 + 0.0096\ii$& 81&   1.0044 &   $\ll q^{-2n} $  \\ \hline 
 Jacobi(1,2)&-5/3 &$\sqrt{2/3}$&300& -1                           &   ---           &1    & $\simeq n^{-2}$   \\ \hline
  \end{tabular}
\caption{Parameters and numerical results for Examples \ref{jumping}, \ref{jumping2}, \ref{jumping3} and \ref{realev}. If $\lambda_\infty\notin\Real$ then $n_0$ denotes the maximal $n$ for which $\lambda_n\in\Real$, $q$ was computed according to Theorem \ref{main-geom}.} \label{Table}
\end{table}

\begin{figure}[htb]
\includegraphics[width=300pt]{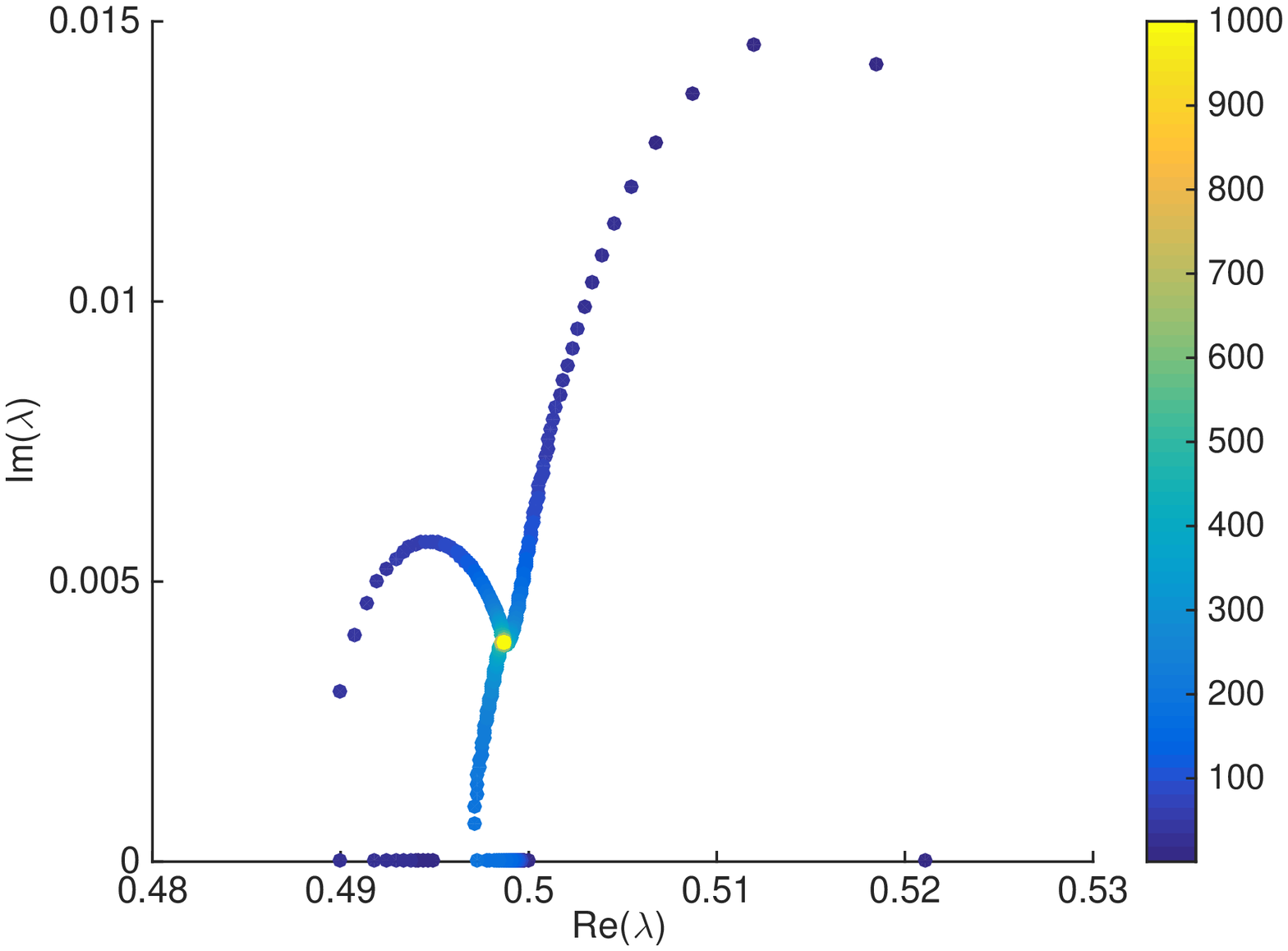} 
\includegraphics[width=300pt]{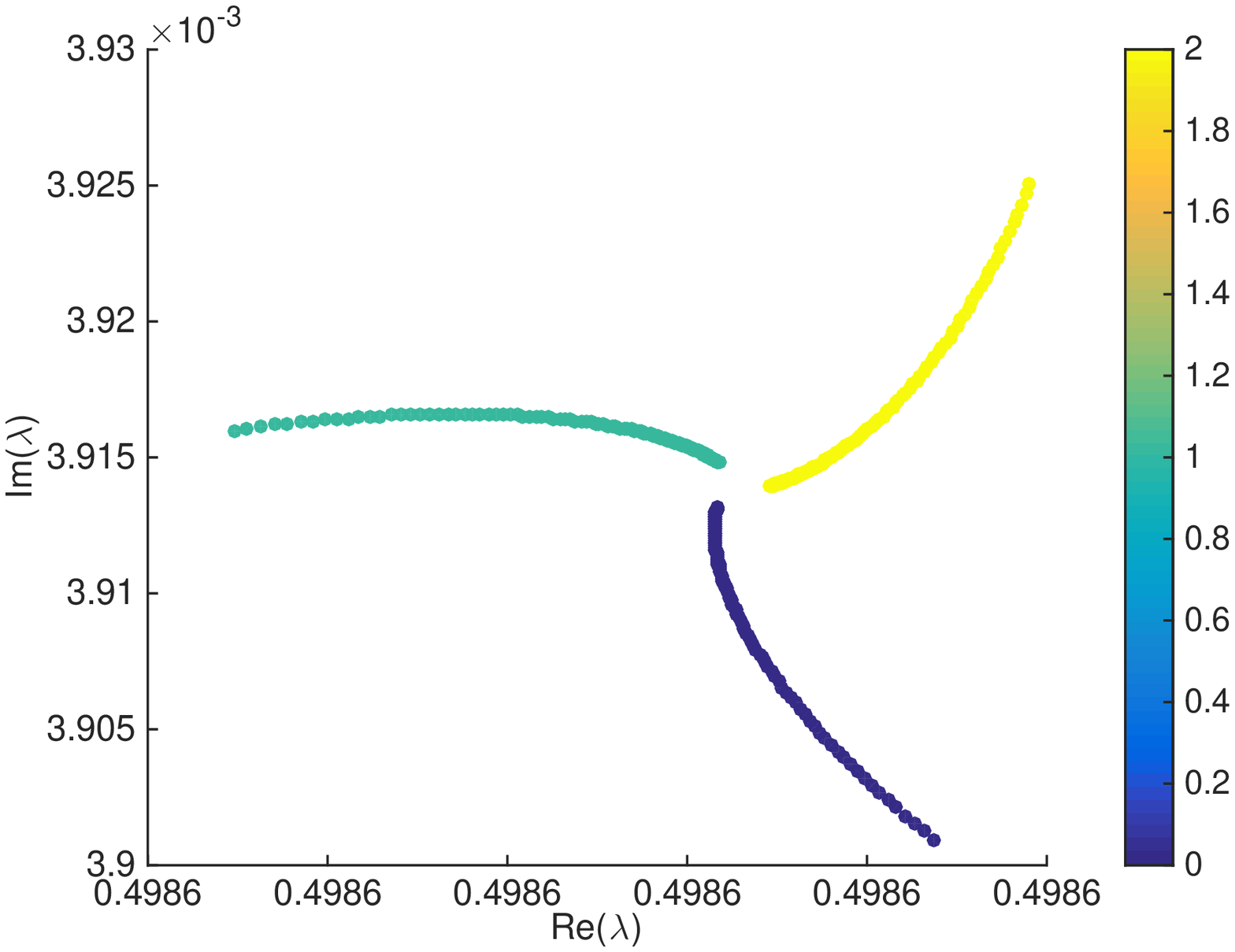}
\caption{The only eigenvalue  $\lambda_n$ of nonpositive type of the matrix $H_{[0,n]}$ given in eq. \eqref{H0n} with $a_0=0.5$, $b_0=0.05$ and $H_{[1,n]}$ being the Jacobi matrix corresponding to the Legendre polynomials.  The points, each corresponding to a different $n$, are plotted on the complex plane color coded according to $n$ in the upper picture, and according to $n \mod 3$  ($0$--blue, $1$--cyan, $2$--yellow) in the lower picture, where we show a zoomed detail around  $\lambda_\infty$. The three branches of $\lambda_n$ mentioned in the text are clearly visible. See Example \ref{jumping}.}\label{Legendre_compl}
\end{figure}

\begin{figure}[htb]
\includegraphics[width=300pt]{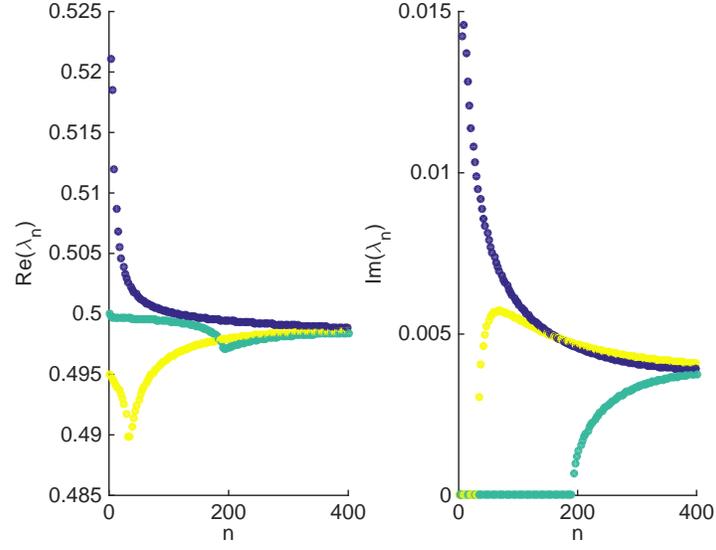}
 \caption{Real and imaginary part of the eigenvalue $\lambda_n$  from Figure \ref{Legendre_compl} Vs. $n$. The points  color coded according to $n \mod 3$. The cusps in the real parts of the yellow and cyan branches correspond to the change of direction when a branch leaves the real line, see Figure \ref{Legendre_compl}. See Example \ref{jumping}.}\label{Legendre_ReIM}
\end{figure}

\begin{figure}[htb]
\includegraphics[width=300pt]{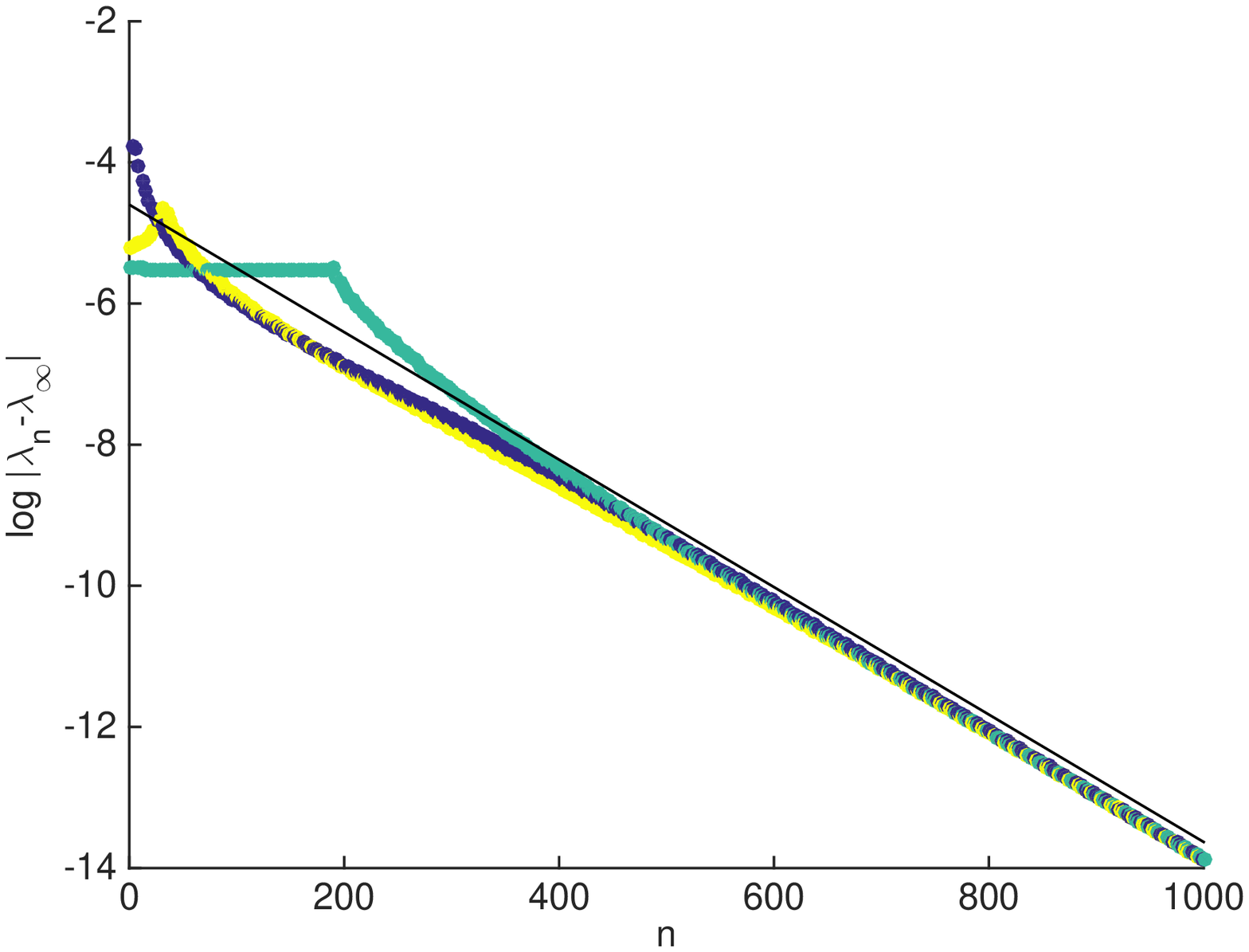}
\caption{Same parameters as in Figure \ref{Legendre_compl}. The black line represents the theoretical bound given in Theorem \ref{main-geom}, shifted down so as to be superimposed on the numerical data. The logarithmic scale on the vertical axis makes evident the exponential convergence of $\lambda_n$  to $\lambda_\infty$. See Example \ref{jumping}.}\label{Legendre_conv}
\end{figure}

\begin{figure}[htb]
\includegraphics[width=300pt]{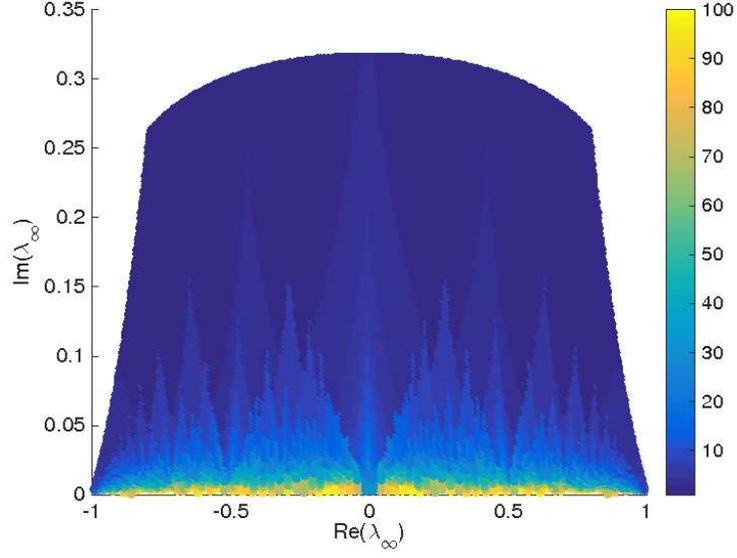}
\caption{The value $n_0$ where the last $\lambda_n$ branch leaves the real axis as a function of $\lambda_\infty$. $H_{[1,n]}$ is the Jacobi matrix corresponding to the Legendre polynomials and $a_0$ and $b_0$ vary to get the different values of $\lambda_\infty$. See Example \ref{jumping}.}\label{Ns}
\end{figure}

 \begin{figure}[htb]
\includegraphics[width=300pt]{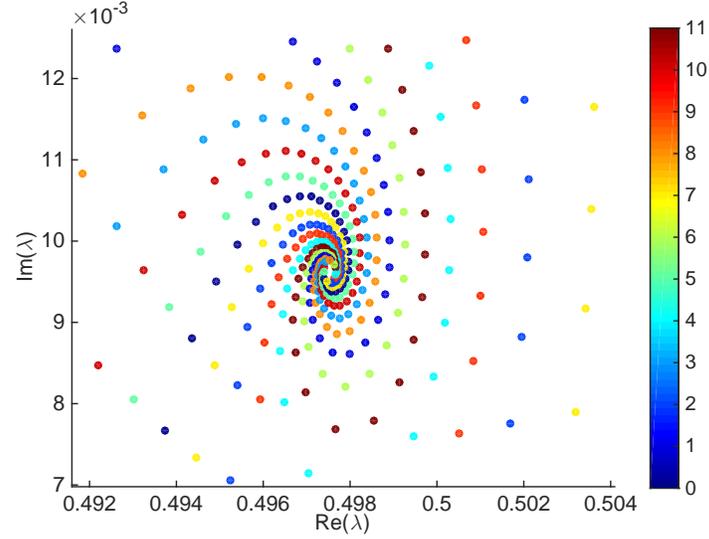}
 \caption{The eigenvalues $\lambda_n$  for $H_{[1,n]}$ associated with the Wigner measure and parameters $a_0=0.5$, $b_0=0.1$, plotted in the complex plane, color coded according to $n \mod 12$ where twelve is the number of $\lambda_n$ branches. See Example \ref{jumping2}.
 }\label{Wigner_compl}
\end{figure}

\begin{figure}[htb]
\includegraphics[width=300pt]{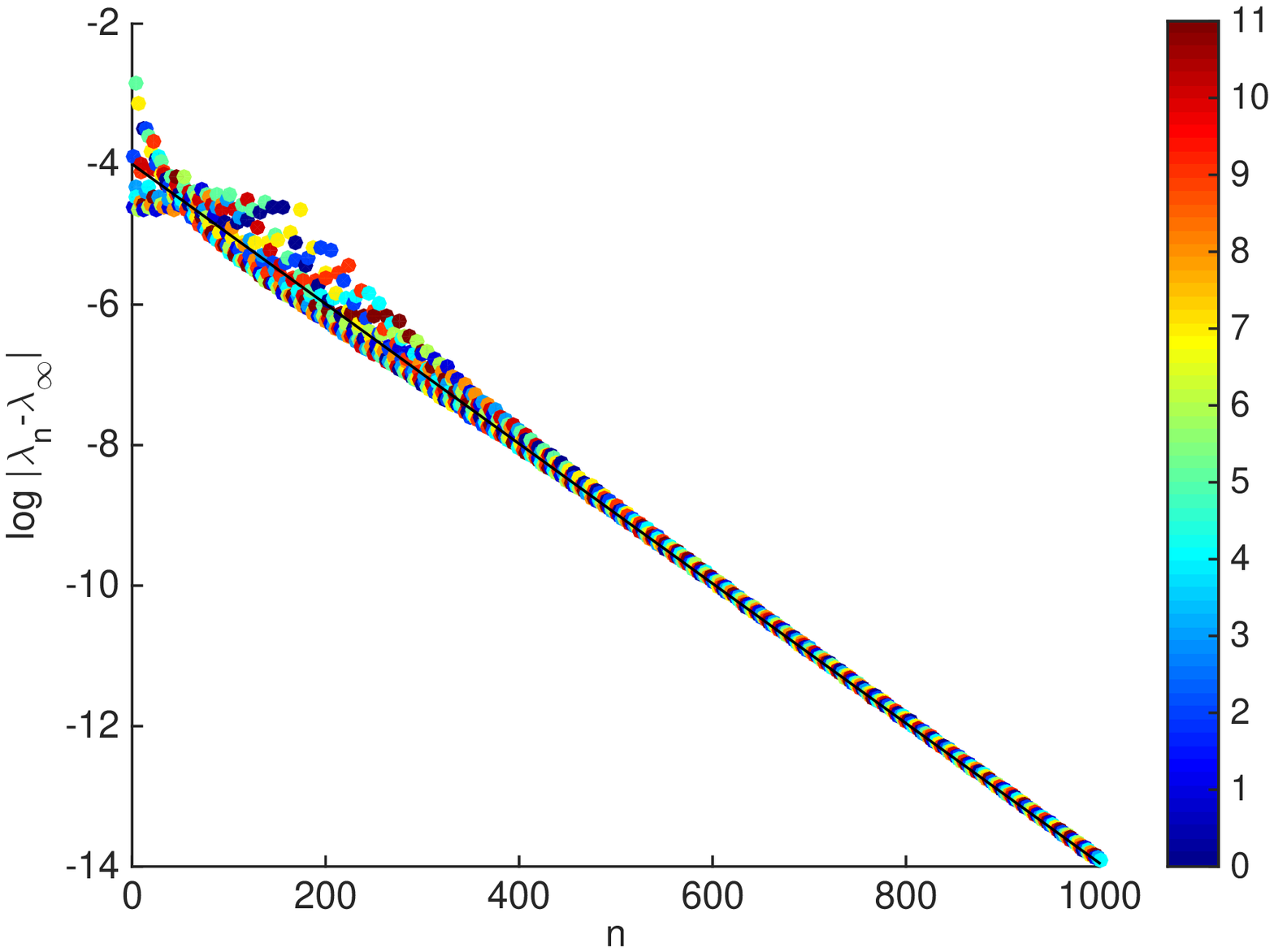} 
\caption{Same parameters and color code as in Figure \ref{Wigner_compl}. The black line represents the theoretical bound given in Theorem \ref{main-geom}, shifted down so as to be superimposed on the numerical data. The logarithmic scale on the vertical axis makes evident the exponential convergence of $\lambda_n$  to $\lambda_\infty$. See Example \ref{jumping2}.}\label{Wigner_conv}
\end{figure}

\begin{figure}[htb]
\includegraphics[width=300pt]{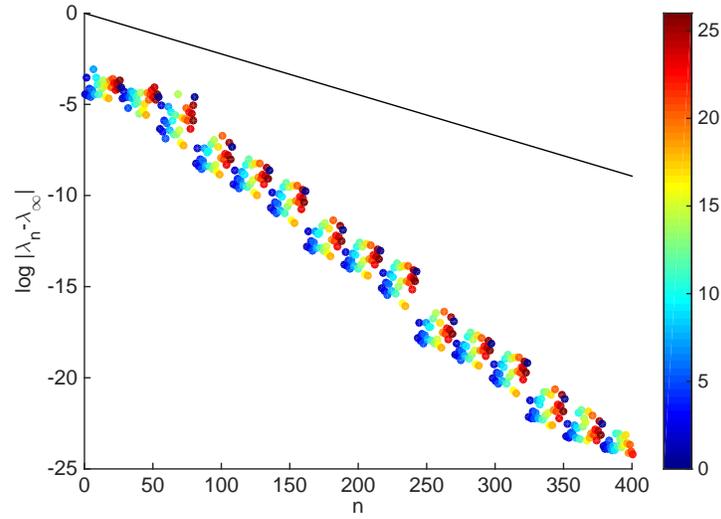}
\caption{Exponential convergence of $\lambda_n$ to $\lambda_\infty$. Here $H_{[1,\infty)}$ is associated with a measure on a Cantor-like set and $a_0=0.5$, $b_0=0.1$.  The black line represents the theoretical bound given in Theorem \ref{main-geom}. Color coding according to $n \mod 27$ evidences a repeated pattern of 27 points.  The observed convergence is again exponential but much faster than the bound given by Theorem \ref{main-geom}: the observed $q$ is approximately $1.1147$, much larger than the theoretical one $1.0044$. See Example \ref{jumping3}.}\label{Cantor_conv}
\end{figure}

\begin{figure}[htb]
\includegraphics[width=300pt]{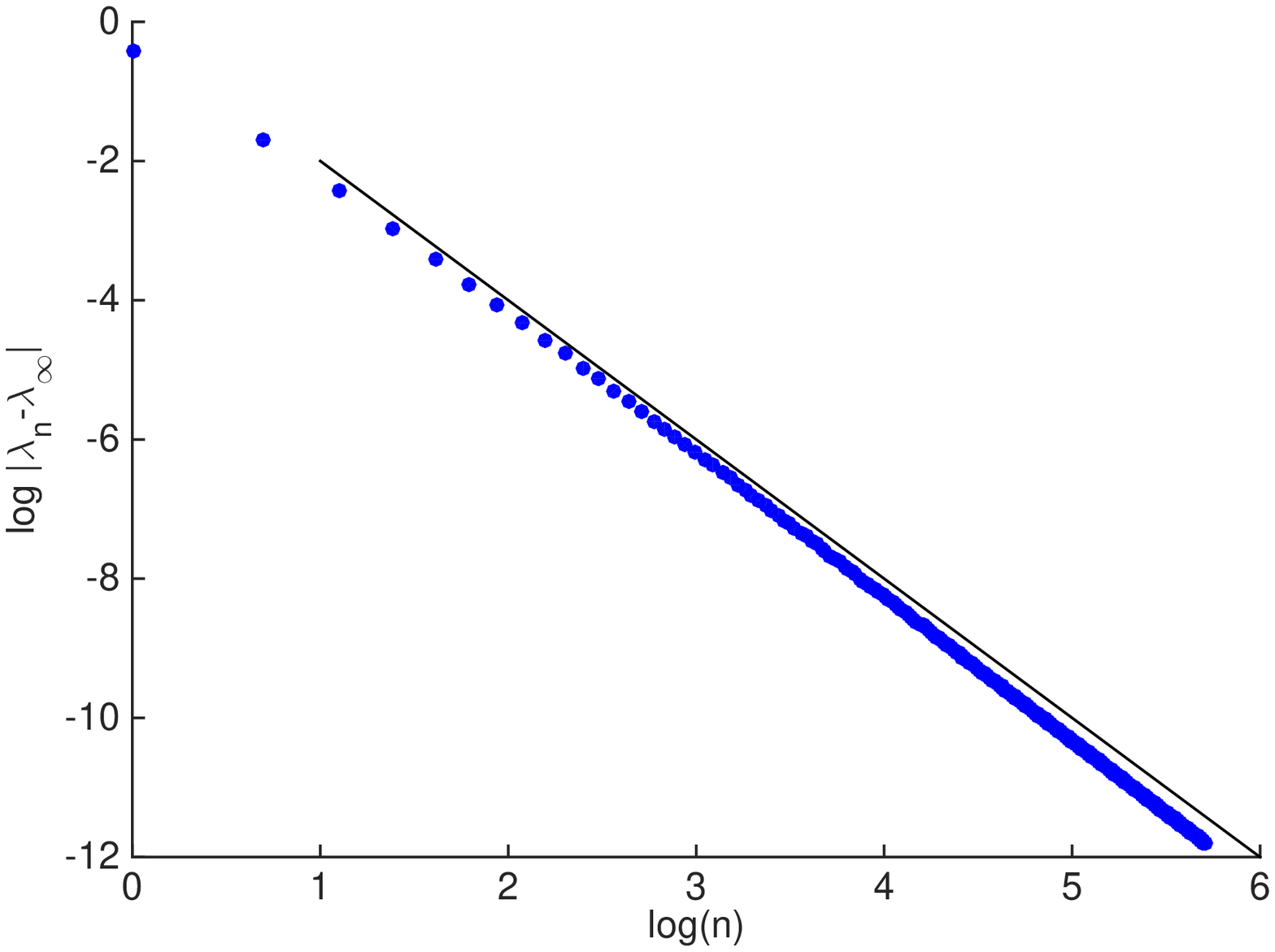}
 \caption{Log-log plot showing $\mathcal{O}(n^{-2})$ convergence of $\lambda_n$ to $\lambda_\infty=-1$ in a case when $\lambda_\infty$ is embedded in the spectrum of $H_{[1,\infty)}$. See Example \ref{realev}.}\label{Jacobicrit_conv}
\end{figure}


\begin{thebibliography}{20}
\bibitem{exep} There are in practice some rare exceptions when a real spectrum is not sought -- as for example when looking for resonances through complex dilatation of a Hamiltonian \cite{Buch} -- and it has to be kept in mind that, even when the Hamiltonian can be made to be Hermitian by symmetrization, different linear combinations of orderings of the operators in the Hamiltonian itself can be possible (see e.g. the case of the double pendulum \cite{per}) corresponding to different physical systems having the same classical Hamiltonian.
\bibitem{note} Under T reflection one replaces $i$ by $−i$ in the operator while under P reflection one replaces $x$ by $a−x$, where $a/2$ is the origin about which one is performing the parity reflection.
\bibitem{S1} The pictures names in \cite{S} are {\sf Legendre\_a\_0.5\_b\_0.xxx\_compl.eps} for $\lambda_n$ on the complex plane, {\sf Legendre\_a\_0.5\_b\_0.xxx\_ReIm.eps} for real and imaginary part of $\lambda_n$,  and {\sf Legendre\_a\_0.5\_b\_0.xxx\_conv.eps} for the logarithmic plot for the convergence rate.
\bibitem{S2} The pictures names in \cite{S} are {\sf Legendre\_a\_1.001\_b\_0.001\_xxxx.eps} for $\lambda_n$ on the complex plane ({\sf xxxx=compl}), real and imaginary part of $\lambda_n$ ({\sf xxxx=ReIm}),and the convergence rate ({\sf xxxx=conv}).
\bibitem{S3} The names are {\sf Wigner\_a\_0.5\_b\_0.1\_xxxx.eps}.
\end{thebibliography}
\end{document}